\documentclass[conference]{IEEEtran}

\usepackage[letterpaper, left=0.680in, right=0.680in, bottom=1.049in, top=0.71in]{geometry}

\usepackage{ifpdf}
\ifCLASSINFOpdf
\usepackage[pdftex]{graphicx}
\graphicspath{{./Figures/}}
\DeclareGraphicsExtensions{.pdf,.jpeg,.png}
\else
\usepackage[dvips]{graphicx}
\DeclareGraphicsExtensions{.pdf}
\usepackage[caption=false, font=footnotesize]{subfig}
\fi

\usepackage{amsthm}
\usepackage{balance}
\usepackage{epstopdf}
\usepackage{amsmath}
\usepackage{amssymb}
\usepackage{color}
\usepackage{array}
\usepackage{mathtools}
\usepackage[ruled, vlined, linesnumbered]{algorithm2e}
\usepackage{enumitem}
\usepackage[mathcal]{euscript}
\usepackage{tabularx}
\usepackage{multirow}
\usepackage{booktabs}
\usepackage{makecell}
\usepackage{aligned-overset}
\usepackage{bbm}
\usepackage{subfig}
\usepackage{url}
\usepackage{tikz}
\usetikzlibrary{shapes.geometric, arrows.meta, positioning, calc}

\newtheorem{lemma}{Lemma}
\newtheorem{theorem}{Theorem}

\newtheorem{corollary}{Corollary}

\newtheorem{remark}{Remark}

\usepackage{xcolor}

\begin{document}

	\title{Communication-Centric RIS-Assisted ISAC: \\ Signal Modeling and BER Analysis}

	\author{
		\IEEEauthorblockA{Yosefine Triwidyastuti and Tri~Nhu~Do}
		\IEEEauthorblockA{Telecom Neural Detection Lab, Department of Electrical Engineering, Polytechnique Montr\'{e}al, Montreal, QC, Canada}
		\IEEEauthorblockA{Emails: yosefine-2.triwidyastuti@polymtl.ca, tri-nhu.do@polymtl.ca}
	}

	\maketitle

\begin{abstract}
We propose and analyze a communication-centric reconfigurable intelligent surface (RIS)-assisted integrated sensing and communication (ISAC) system, in which a monostatic radar simultaneously senses a moving target and serves a user equipment (UE) over Nakagami-$m$ fading. We design a dual-function phase-modulated continuous-wave (PMCW) waveform that embeds the data stream directly into the radar pulse train: each pulse carries one full maximum-length sequence whose polarity is flipped by a binary phase-shift keying data symbol, so that the same emission preserves the sharp range autocorrelation required for sensing while conveying one bit per pulse to the UE. We further propose a communication-centric RIS phase configuration that co-phases each element onto the direct radar-to-UE path, yielding a coherent superposition at the UE and a received-power gain that scales with the square of the number of elements. We show that from the radar's perspective, however, the same surface behaves as an uncontrolled scatterer, since the resulting reflection paths are mis-phased and do not benefit from array combining. We derive a closed-form approximation for the average UE bit error rate based on a moment-matched Gamma approximation, and we show that the same waveform still forms a usable range--Doppler map for sensing. Monte-Carlo simulations corroborate the analytical results.
\end{abstract}


\section{Introduction}

Integrated sensing and communication (ISAC) jointly performs radar sensing and data communication over shared wireless resources, improving spectral and energy efficiency while reducing hardware complexity and cost \cite{BouzabiaTVT2025}.
In parallel, reconfigurable intelligent surfaces (RISs) have emerged as a promising means of controlling wireless propagation \cite{DoTCOM2021}. By adjusting the phase shifts of passive reflecting elements, an RIS can shape the channel to support both communication and sensing. This motivates the RIS-ISAC paradigm, where waveform design, channel-phase optimization, and RIS configuration are jointly considered \cite{JoyMCOMSTD2025}. However, because a configuration optimized for one function may degrade the other, the resulting sensing--communication tradeoff must be carefully characterized.

In this work, we consider a monostatic RIS-assisted ISAC system employing a phase-modulated continuous-wave (PMCW) waveform \cite{BanerjeeOJCOMS2026}. This setting couples two design questions: $(i)$ dual-function waveform design, where a maximum-length-sequence (MLS)-based PMCW waveform must preserve the range-domain matched-filter gain required for sensing while carrying communication data; and $(ii)$ RIS phase configuration under a multiple-reflection channel model, where the same RIS generates the coherent communication cascade together with the singly RIS-assisted target return, the static RIS self-return, and the doubly RIS-assisted target return, each with distinct delays, Doppler shifts, and residual phases. Aligned with the concept of dual-function radar-communication (DFRC) approaches \cite{VaeziCOMST2026}, the proposed MLS-based PMCW waveform can directly embed binary data through BPSK modulation of each MLS sequence, preserving sharp range autocorrelation while avoiding the overhead of additional modulation layers. While dual-function waveforms and RIS phase optimization have each been studied separately \cite{JoyMCOMSTD2025,BanerjeeOJCOMS2026}, the effect of a purely communication-centric RIS configuration on the coexisting sensing operation has not been fully characterized.
Motivated by this gap, we propose and analyze a communication-centric (CC) RIS-ISAC system. The main contributions are summarized as follows.

\begin{itemize}
    \item We develop a monostatic PMCW RIS-ISAC framework in which a single MLS-based dual-function waveform conveys one binary phase-shift keying (BPSK) bit per pulse to the UE while remaining usable for range--Doppler sensing at the radar, and we derive the corresponding sensing and communication received-signal models.

    \item We propose a communication-centric RIS phase configuration that co-phases the RIS-assisted cascade onto the direct radar--UE path. We show that this provides the UE with an asymptotic received-power gain proportional to $L^2$ (i.e., $20\log_{10}L$ dB relative to a single-element reference), where $L$ is the number of RIS elements.

    \item We derive a closed-form approximation to the average communication BER over Nakagami-$m$ fading using a moment-matched Gamma model and the moment-generating function (MGF), together with an exact no-RIS Gamma baseline.

    \item In addition, we show that, under the communication-centric configuration, the direct target echo does not interact with the RIS; consequently, the RIS-assisted returns are mis-phased at the radar and do not enhance the sensing SNR, while the self-return adds zero-Doppler clutter. The sensing branch therefore performs essentially as in the no-RIS case. 
\end{itemize}

\section{CC-RIS-ISAC System and Signal Model}

We consider a communication-centric RIS-assisted ISAC system in which a single monostatic radar node reuses one PMCW waveform to simultaneously $(i)$ sense a moving point target and $(ii)$ convey data to a UE. The link is assisted by an RIS whose phase shifts are configured solely to enhance the communication path.

\subsection{Network Topology and RIS Model}

\begin{figure}[t]
    \centering
    \begin{tikzpicture}[
        >={Latex[length=1.8mm]},
        font=\footnotesize,
        comm/.style={dashed,semithick,blue!55!black},
        sens/.style={dotted,thick,purple!55!black},
        radarN/.style={regular polygon,regular polygon sides=3,draw=black,
                       fill=blue!60!white,minimum size=8pt,inner sep=0pt},
        risN/.style={rectangle,draw=black,fill=green!60!black,
                     minimum size=8pt,inner sep=0pt},
        ueN/.style={circle,draw=black,fill=orange!85!black,
                    minimum size=8pt,inner sep=0pt},
        tgtN/.style={star,star points=5,draw=black,fill=purple!65,
                     minimum size=12pt,inner sep=0pt},
        lbl/.style={font=\scriptsize,inner sep=1.5pt}
    ]
        \coordinate (R) at (0,0);
        \coordinate (I) at (2.10,0.42);
        \coordinate (U) at (4.20,0);
        \coordinate (T) at (7.00,2.80);

        \draw[comm] (R) -- node[lbl,below]                 {$h_{\mathsf{RU}},\,d_{\mathsf{RU}}$} (U);
        \draw[comm] (R) -- node[lbl,above left,pos=.55]    {$h_{\mathsf{RI}},\,d_{\mathsf{RI}}$} (I);
        \draw[comm] (I) -- node[lbl,above right,pos=.45]   {$h_{\mathsf{IU}},\,d_{\mathsf{IU}}$} (U);
        \draw[sens] (R) -- node[lbl,sloped,above,pos=.45]  {$h_{\mathsf{RT}},\,d_{\mathsf{RT}}$} (T);
        \draw[sens] (I) -- node[lbl,sloped,below,pos=.58]  {$h_{\mathsf{IT}},\,d_{\mathsf{IT}}$} (T);

        \node[radarN] (Rn) at (R){};
        \node[risN]   (In) at (I){};
        \node[ueN]    (Un) at (U){};
        \node[tgtN]   (Tn) at (T){};

        \node[below=2pt of Rn]{$\mathsf{R}$ (radar)};
        \node[below=2pt of Un]{$\mathsf{U}$ (UE)};
        \node[below right=-1pt and 1pt of In]{$\mathsf{I}$};
        \node[above left=0pt and 1pt of Tn]{$\mathsf{T}$ (target, $\sigma$)};

        \draw[->,very thick] (Tn) -- ++(0.85,-0.85)
              node[below,font=\scriptsize]{$\mathbf v_{\mathsf T}$};

        \node[anchor=west,font=\scriptsize] (cfg) at (-0.15,2.55)
              {$\boldsymbol{\Phi}^{\star}=\mathrm{diag}\!\big(\kappa_l e^{j\vartheta_l^{\star}}\big),\ l=1,\dots,L$};
        \draw[->,gray,shorten >=2pt] (cfg.south) to[bend right=12] (In.north);

        \begin{scope}[shift={(-0.15,-0.95)}]
            \draw[comm] (0,0) -- (0.55,0);
            \node[right=1pt,font=\scriptsize] at (0.55,0){communication link};
            \draw[sens] (3.55,0) -- (4.10,0);
            \node[right=1pt,font=\scriptsize] at (4.10,0){sensing link};
        \end{scope}
    \end{tikzpicture}
    \caption{Network topology of the CC-RIS-ISAC system.}
    \label{fig_topology}
\end{figure}
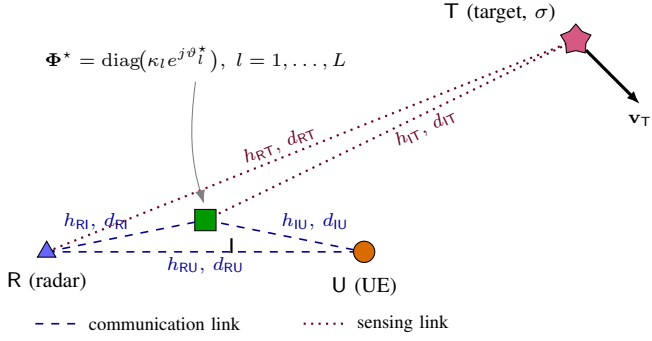

Four nodes are deployed in a two-dimensional plane, as illustrated in Fig.~\ref{fig_topology}: the radar $\mathsf{R}$, the UE $\mathsf{U}$, the moving target $\mathsf{T}$, and the RIS $\mathsf{I}$. The location of node $\mathsf{A}\in\{\mathsf{R},\mathsf{T},\mathsf{U},\mathsf{I}\}$ is $(x_{\mathsf A},y_{\mathsf A})$, and $d_{\mathsf{AB}}$ denotes the distance between nodes $\mathsf A$ and $\mathsf B$. The radar is a co-located monostatic transceiver; the target is a point scatterer with velocity $\mathbf v_{\mathsf T}$; and the RIS has $L$ passive elements indexed by $l=1,\ldots,L$. The complex reflection coefficient of element $l$ is
\begin{align}
    \theta_l = \kappa_l\, e^{j\vartheta_l},
    \qquad
    \kappa_l\in(0,1],\quad \vartheta_l\in[0,2\pi),
    \label{eq_ris_element}
\end{align}
with $\boldsymbol{\Phi}=\operatorname{diag}([\kappa_1 e^{j\vartheta_1},\ldots,\kappa_{L}e^{j\vartheta_{L}}])$, where $\kappa_l$ is the amplitude reflection coefficient and $\vartheta_l$ the controllable phase shift~\cite{Triwidyastuti_Access_2025}.

\subsection{Dual-Function Waveform and Frame Structure}

We employ a maximum-length sequence (MLS) of length $N=2^k-1$, generated by a primitive polynomial of degree $k$. After bipolar mapping, $\mathbf s=[s_1,\ldots,s_{N}]^{\mathsf T}\in\{\pm1\}^N$, extended $N$-periodically ($s_{n+N}=s_n$) so that circular indices of the form $s[(n-k)\bmod N]$ are well defined. Within each coherent processing interval (CPI), $M$ information bits $b_m\in\{0,1\}$, $m=1,\ldots,M$, are mapped to BPSK symbols $d_m=1-2b_m\in\{\pm1\}$. Each bit modulates one complete MLS pulse, so the transmitted fast-time code vector during pulse $m$ is $d_m\,\mathbf s$. The same emission therefore supports both sensing and communication.

Each PMCW pulse contains $N$ chips of duration $T_c=1/f_s$, with fast-time index $n=1,\ldots,N$ resolving delay (range) and slow-time index $m=1,\ldots,M$ resolving Doppler while carrying data. With contiguous code repetitions, $T_{\mathrm{PRI}}=NT_c$ and $T_{\mathrm{CPI}}=MNT_c$. Letting $p(t)$ be a unit-energy chip pulse with chip energy $E_c$, the transmitted baseband signal is
\begin{align}
    x_{\mathrm{BB}}(t)
    =
    \sqrt{E_c}
    \sum_{m=1}^{M}d_m
    \sum_{n=1}^{N}s_n\,
    p\!\left(t-nT_c-mT_{\mathrm{PRI}}\right).
\end{align}
Since one bit spans all $N$ chips, the bit energy is $E_b=NE_c$. The nominal direct-path range/velocity resolutions and unambiguous spans are $\Delta R\simeq c_0T_c/2$, $R_{\max}=c_0T_{\mathrm{PRI}}/2$, $\Delta v=\lambda/(2MT_{\mathrm{PRI}})$, and $v_{\max}=\lambda/(4T_{\mathrm{PRI}})$, with $\lambda=c_0/f_c$.

\subsection{RIS-assisted ISAC Channel Models}

The system comprises five links: $\mathsf R\!\leftrightarrow\!\mathsf U$, $\mathsf R\!\leftrightarrow\!\mathsf I$, $\mathsf I\!\leftrightarrow\!\mathsf U$ (communication), and $\mathsf R\!\leftrightarrow\!\mathsf T$, $\mathsf I\!\leftrightarrow\!\mathsf T$ (sensing).

\subsubsection{Communication links}
A generic communication channel between nodes $\mathsf A$ and $\mathsf B$ is
\begin{align}
    h_{\mathsf{AB}}
    =
    \sqrt{\beta_{\mathsf{AB}}}\,a_{\mathsf{AB}}e^{j\phi_{\mathsf{AB}}},
    \quad
    \beta_{\mathsf{AB}}=\Psi_{\mathsf{AB}}(d_{\mathsf{AB}}/d_0)^{-\alpha_{\mathsf{AB}}},
\end{align}
where $\beta_{\mathsf{AB}}$ is the one-way large-scale gain, $a_{\mathsf{AB}}\sim\operatorname{Nakagami}(m_{\mathsf{AB}},1)$ the fading amplitude, and $\phi_{\mathsf{AB}}\sim\mathcal U[0,2\pi)$ the fading phase, independent of $a_{\mathsf{AB}}$ \cite[Eq.~(3.13)]{Goldsmith2005}. This applies to the direct channel $h_{\mathsf{RU}}$ and, per element $l$, to $h_{\mathsf{RI},l}$ and $h_{\mathsf{IU},l}$. Under the far-field approximation, $\beta_{\mathsf{RI},l}\!\simeq\!\beta_{\mathsf{RI}}$, $\beta_{\mathsf{IU},l}\!\simeq\!\beta_{\mathsf{IU}}$, and reciprocity gives $h_{\mathsf{IR},l}=h_{\mathsf{RI},l}$.

\subsubsection{Sensing links}
The radar--target and RIS--target links are modeled as deterministic geometric propagation channels following the radar range equation~\cite{Richards2014}; the target's fluctuating scattering is carried separately by $\rho_{\mathsf T}=\sqrt{\sigma}\,e^{j\phi_\sigma}$.
$\sigma$ is the radar cross section, which depends on the roughness, size, and shape of the scatterer \cite[cf.~Eq.~(2.25)]{Goldsmith2005}.
The one-way coefficient is $h_{\mathsf{AT}}=\sqrt{\beta_{\mathsf{AT}}}\,e^{-j\frac{2\pi}{\lambda}d_{\mathsf{AT}}}$,
where the phase shift $e^{-j2\pi d_{\mathsf{AT}}/\lambda}$ is due to the distance $d_{\mathsf{AT}}$ the wave travels~\cite[Eq.~(2.6)]{Goldsmith2005} and $\beta_{\mathsf{AT}}$ follows a log-distance model with range-dependent atmospheric and system losses~\cite{Richards2014}. 
Based on the radar range equation, under reciprocity, the monostatic radar--target--radar coefficient is
\begin{align}
    \xi_0
    =
    \rho_{\mathsf T}\,h_{\mathsf{RT}}^2
    =
    \rho_{\mathsf T}\,\beta_{\mathsf{RT}}\,e^{-j\frac{4\pi}{\lambda}d_{\mathsf{RT}}},
    \quad
    |\xi_0|^2=\sigma\beta_{\mathsf{RT}}^2.
    \label{eq_direct_target_coefficient}
\end{align}

\subsection{Communication-Centric RIS Configuration}

The RIS phases are selected to coherently combine the direct and RIS-assisted communication components at the UE:
\begin{align}
    \vartheta_l^\star
    =
    \phi_{\mathsf{RU}}
    -\phi_{\mathsf{RI},l}
    -\phi_{\mathsf{IU},l}
    \pmod{2\pi},
    \quad
    \theta_l^\star=\kappa_l e^{j\vartheta_l^\star}.
    \label{eq_comm_ris_phase}
\end{align}
This choice aligns every RIS-assisted path with the direct path at the UE, as the following result formalizes.

\begin{lemma}[Coherent communication combining]
\label{lem_eff_amp}
Under the communication-centric phase configuration \eqref{eq_comm_ris_phase} and the far-field approximation, the effective communication channel $h_{\mathrm{eff}}=h_{\mathsf{RU}}+\sum_{l=1}^{L}h_{\mathsf{IU},l}\theta_l^\star h_{\mathsf{RI},l}$ factors as $h_{\mathrm{eff}}=e^{j\phi_{\mathsf{RU}}}A_c$, where the effective amplitude
\begin{align}
    A_c
    \triangleq
    |h_{\mathrm{eff}}|
    =
    \sqrt{\beta_{\mathsf{RU}}}\,a_{\mathsf{RU}}
    +
    \sqrt{\beta_{\mathsf{RI}}\beta_{\mathsf{IU}}}
    \sum_{l=1}^{L}
    \kappa_l\,a_{\mathsf{RI},l}a_{\mathsf{IU},l}
    \label{eq_effective_amplitude}
\end{align}
is a real, nonnegative, phase-aligned superposition of the direct and the $L$ RIS-assisted paths.
\end{lemma}
\begin{proof}
Substituting $h_{\mathsf{RU}}=\sqrt{\beta_{\mathsf{RU}}}\,a_{\mathsf{RU}}e^{j\phi_{\mathsf{RU}}}$, $h_{\mathsf{RI},l}=\sqrt{\beta_{\mathsf{RI},l}}\,a_{\mathsf{RI},l}e^{j\phi_{\mathsf{RI},l}}$, $h_{\mathsf{IU},l}=\sqrt{\beta_{\mathsf{IU},l}}\,a_{\mathsf{IU},l}e^{j\phi_{\mathsf{IU},l}}$, and $\theta_l^\star=\kappa_l e^{j\vartheta_l^\star}$, the phase of the $l$-th cascaded term is $\phi_{\mathsf{RI},l}+\vartheta_l^\star+\phi_{\mathsf{IU},l}$. By \eqref{eq_comm_ris_phase} this equals $\phi_{\mathsf{RU}}$ for every $l$. Factoring the common phase $e^{j\phi_{\mathsf{RU}}}$ and applying $\beta_{\mathsf{RI},l}\simeq\beta_{\mathsf{RI}}$, $\beta_{\mathsf{IU},l}\simeq\beta_{\mathsf{IU}}$ yields \eqref{eq_effective_amplitude}; all amplitudes are nonnegative, so the terms add coherently.
\end{proof}

The coherent amplitude addition in \eqref{eq_effective_amplitude} directly determines how the received power scales with the array size.

\begin{corollary}[$L^2$ array gain]
\label{cor_l2_gain}
With identical statistics across elements ($\kappa_l=\kappa$, $m_{\mathsf{RI}},m_{\mathsf{IU}}$ fixed), the mean received power satisfies $\mathbb E[|h_{\mathrm{eff}}|^2]=\Theta(L^2)$, i.e., a gain of $20\log_{10}L$ dB relative to a single-element reference.
\end{corollary}
\begin{proof}
By Lemma~\ref{lem_eff_amp} and linearity of expectation, $\mathbb E[A_c]=\sqrt{\beta_{\mathsf{RU}}}\,\mathbb E[a_{\mathsf{RU}}]+L\kappa\sqrt{\beta_{\mathsf{RI}}\beta_{\mathsf{IU}}}\,\mathbb E[a_{\mathsf{RI}}]\mathbb E[a_{\mathsf{IU}}]$, using independence of $a_{\mathsf{RI},l}$ and $a_{\mathsf{IU},l}$. The RIS term grows linearly in $L$ with a strictly positive mean (Nakagami amplitudes have $\mathbb E[a]>0$), so $\mathbb E[A_c]=\Theta(L)$ and $\mathbb E[A_c^2]\geq(\mathbb E[A_c])^2=\Theta(L^2)$. Conversely $\mathbb E[A_c^2]=\mathcal O(L^2)$ from \eqref{eq_moment_recursion}, hence $\mathbb E[A_c^2]=\Theta(L^2)$.
\end{proof}

\subsection{Proposed Received Signal Modelling}

After unit-energy chip-matched filtering, the noise samples at the UE and radar are i.i.d. $\mathcal{CN}(0,N_0)$. With the one-way radar-to-UE communication delay $\tau_{\mathsf{RU}}=d_{\mathsf{RU}}/c_0$ and its delay bin $\ell_u=\lfloor\tau_{\mathsf{RU}}/T_c\rceil+1$,
the UE chip sample is
\begin{align}
    r_{\mathrm{ue}}[n,m]
    =
    \sqrt{E_c}\,h_{\mathrm{eff}}\,
    s[(n-\ell_u+1)\bmod N]\,d_m
    +
    z_{\mathrm{ue}}[n,m].
    \label{eq_r_ue_final}
\end{align}
At the radar, the monostatic receiver collects four returns, indexed $i=0,\ldots,3$: $(0)$ the direct target return; $(1)$ the singly RIS-assisted target return; $(2)$ the static RIS self-return; and $(3)$ the doubly RIS-assisted target return.

\begin{lemma}[Effective radar return coefficients]
\label{lem_returns}
Under channel reciprocity and the stop-and-hop approximation $|f_{D,i}|NT_c\ll1$, the monostatic radar observation is
\begin{align}
    r_{\mathrm{rad}}[n,\!m]
    \!\!=\!\!&\;
    \sqrt{E_c} \!
    \sum_{i=0}^{3} \!
    \xi_i\,
    s[(n-\ell_i+1)\!\bmod\! N]\,d_m
    e^{j2\pi f_{D,i}mT_{\mathrm{PRI}}}
    \nonumber\\
    &+
    z_{\mathrm{rad}}[n,m],
    \label{eq_radar_echo}
\end{align}
where, with the one-way RIS-assisted cascade $q_{\mathsf{RIT}}=\sum_l\theta_l^\star h_{\mathsf{RI},l}h_{\mathsf{IT},l}$, the effective coefficients are
\begin{align}
    \xi_0 &= \rho_{\mathsf T}\beta_{\mathsf{RT}}e^{-j\frac{4\pi}{\lambda}d_{\mathsf{RT}}},
    & \xi_1 &= 2\rho_{\mathsf T}h_{\mathsf{RT}}q_{\mathsf{RIT}},\nonumber\\
    \xi_2 &= \textstyle\sum_l\theta_l^\star h_{\mathsf{RI},l}^2,
    & \xi_3 &= \rho_{\mathsf T}q_{\mathsf{RIT}}^2,
    \label{eq_four_returns}
\end{align}
with delays $\tau_0=\tfrac{2d_{\mathsf{RT}}}{c_0}$, $\tau_1=\tfrac{d_{\mathsf{RI}}+d_{\mathsf{IT}}+d_{\mathsf{RT}}}{c_0}$, $\tau_2=\tfrac{2d_{\mathsf{RI}}}{c_0}$, $\tau_3=\tfrac{2(d_{\mathsf{RI}}+d_{\mathsf{IT}})}{c_0}$, and Doppler shifts $f_{D,0}=-\tfrac{2v_r^{\mathsf R}}{\lambda}$, $f_{D,1}=-\tfrac{v_r^{\mathsf R}+v_r^{\mathsf I}}{\lambda}$, $f_{D,2}=0$, $f_{D,3}=-\tfrac{2v_r^{\mathsf I}}{\lambda}$, where $v_r^X=\mathbf v_{\mathsf T}^{\mathsf T}\mathbf u_{X\rightarrow\mathsf T}$ is the radial velocity toward node $X\in\{\mathsf R,\mathsf I\}$.
\end{lemma}
\begin{proof}
Each return is the product of the one-way coefficients along its path together with the target scattering coefficient $\rho_{\mathsf T}$ where the target is traversed. For return $(0)$, reciprocity gives $h_{\mathsf{TR}}=h_{\mathsf{RT}}$, so $\xi_0=\rho_{\mathsf T}h_{\mathsf{RT}}^2=\rho_{\mathsf T}\beta_{\mathsf{RT}}e^{-j4\pi d_{\mathsf{RT}}/\lambda}$. The two singly RIS-assisted paths $\mathsf R\!\to\!\mathsf I\!\to\!\mathsf T\!\to\!\mathsf R$ and $\mathsf R\!\to\!\mathsf T\!\to\!\mathsf I\!\to\!\mathsf R$ are reciprocal and share the same delay, Doppler, and scalar coefficient $\rho_{\mathsf T}h_{\mathsf{RT}}q_{\mathsf{RIT}}$, hence add to $\xi_1=2\rho_{\mathsf T}h_{\mathsf{RT}}q_{\mathsf{RIT}}$. The self-return traverses each $h_{\mathsf{RI},l}$ twice, giving $\xi_2=\sum_l\theta_l^\star h_{\mathsf{RI},l}^2$ (no target, no $\rho_{\mathsf T}$). The doubly RIS-assisted path traverses the cascade in both directions, giving $\xi_3=\rho_{\mathsf T}q_{\mathsf{RIT}}^2$. Delays follow from the summed geometric path lengths divided by $c_0$, and each one-way bounce off a moving target contributes a Doppler $-v_r^X/\lambda$ per node $X$ in the path. Under $|f_{D,i}|NT_c\ll1$, the Doppler phase is constant within a pulse.
\end{proof}

The communication-centric phases that benefit the UE have a markedly different effect on the radar returns, as the next remark makes precise.

\begin{remark}[RIS is non-coherent for sensing]
\label{rem_noncoherent}
Substituting \eqref{eq_comm_ris_phase}, the residual sensing phase of the $l$-th cascade term is $\psi_{\mathsf{IT},l}=\phi_{\mathsf{RU}}-\phi_{\mathsf{IU},l}-\frac{2\pi}{\lambda}d_{\mathsf{IT},l}$, which is generally not aligned across $l$. Hence the RIS-assisted returns $\xi_1,\xi_3$ and the self-return $\xi_2$ combine non-coherently, in contrast to the coherent gain of Lemma~\ref{lem_eff_amp} at the UE.
\end{remark}


\section{Proposed Receiver Processing}
\label{sec_receiver}

We consider a dual-branch chain: a sensing branch at the radar and a communication branch at the UE.

\subsection{Sensing Function: Range--Doppler Processing}

Since the radar knows the transmitted BPSK symbols, it removes them via $y[n,m]=d_m^{*}r_{\mathrm{rad}}[n,m]$; as $d_m^*=d_m$ and $|d_m|^2=1$, the noise statistics are unchanged. Each pulse is then circularly correlated with the known MLS,
\begin{align}
    g[k,m]
    =
    \frac1N
    \sum_{n=1}^{N}
    y[n,m]\,s[(n-k+1)\bmod N],
\end{align}
whose normalized periodic autocorrelation is $R_s[\delta]=1$ for $\delta=0$ and $-1/N$ otherwise. A return with delay $\ell_i$ thus produces a peak at $k=\ell_i$, with post-correlation noise variance $N_0/N$. The delay maps to the equivalent range $R_k=c_0(k-1)T_c/2$, so that bin $k=1$ corresponds to zero range, giving range spacing $\Delta R=c_0T_c/2$.

After an optional slow-time window $w[m]$, an $M$-point slow-time DFT with coherent-gain normalization is applied per range bin,
\begin{align}
    G[k,p]
    =
    \frac{1}{\sum_{m=1}^{M}w[m]}
    \sum_{m=1}^{M}
    w[m]\,g[k,m]\,e^{-j2\pi pm/M},
\end{align}
yielding the range--Doppler power map $P_{\mathrm{RD}}[k,p]=|G[k,p]|^2$, with Doppler resolution $\Delta f_D=1/(MT_{\mathrm{PRI}})$. Targets are extracted by a two-dimensional cell-averaging constant false-alarm rate (CFAR) detector~\cite{Barkat2005}. For the direct monostatic return, $f_{D,0}=-2v_r^{\mathsf R}/\lambda$, so the detected Doppler bin $\widehat p$ gives $\widehat v_r^{\mathsf R}=-\lambda\widehat p/(2MT_{\mathrm{PRI}})$.

\subsection{Communication Function: Coherent BPSK Detection}

\emph{Code-delay synchronization.}
The UE first acquires the direct-path code delay $\ell_u$ by maximizing the noncoherent per-bin correlation energy accumulated over the CPI,
\begin{align}
    \widehat\ell_u
    =
    \arg\max_{1\leq k\leq N}
    \sum_{m=1}^{M}
    \left|
    \sum_{n=1}^{N}
    r_{\mathrm{ue}}[n,m]\,s^{*}[(n-k+1)\bmod N]
    \right|^2,
    \label{eq_delay_estimate}
\end{align}
where squaring before the slow-time sum removes the unknown data symbol $d_m$ and channel phase. 
From \eqref{eq_r_ue_final}, with the estimated delay $\widehat\ell_u$ the UE despreads each pulse with the delay-aligned MLS,
\begin{align}
    u_m
    &=
    \frac1N
    \sum_{n=1}^{N}
    r_{\mathrm{ue}}[n,m]\,
    s^{*}[(n-\widehat\ell_u+1)\bmod N]
    \nonumber\\
    &=
    \sqrt{E_c}\,h_{\mathrm{eff}}\,d_m+\nu_m,
\end{align}
with $\nu_m\sim\mathcal{CN}(0,N_0/N)$. Using a prefix of $M_{\mathrm{pl}}$ known pilot symbols $d_m$, $m\in\mathcal P$ with $|\mathcal P|=M_{\mathrm{pl}}$, the effective channel is estimated by least squares,
\begin{align}
    \widehat h_{\mathrm{eff}}
    =
    \frac{\sum_{m\in\mathcal P}d_m^{*}u_m}{\sqrt{E_c}\,M_{\mathrm{pl}}}.
\end{align}
Coherent detection then forms the decision variable $\Lambda_m=\operatorname{Re}\{\widehat h_{\mathrm{eff}}^{*}u_m\}$ and decides $\widehat d_m=+1$ ($\widehat b_m=0$) if $\Lambda_m\geq0$ and $\widehat d_m=-1$ ($\widehat b_m=1$) otherwise.

\section{Performance Analysis}
\label{sec_performance_analysis}

\subsection{Integrated Sensing SNR}

Let $w[m]$ be the slow-time window, with coherent processing efficiency
\begin{align}
    L_w
    \triangleq
    \frac{\left|\sum_{m=1}^{M}w[m]\right|^2}
         {M\sum_{m=1}^{M}|w[m]|^2}\in(0,1],
\end{align}
and $L_w=1$ for a rectangular window. We first quantify the sensing gain delivered by the matched-filter and Doppler-processing chain.

\begin{lemma}[Instantaneous integrated sensing SNR]
\label{lem_sensing_snr}
For the direct target return, the instantaneous SNR at the output of the length-$N$ range correlator and the normalized $M$-point slow-time DFT, evaluated at an on-grid range--Doppler cell, is
\begin{align}
    \chi_0
    =
    NM L_w\frac{E_c}{N_0}|\xi_0|^2
    =
    NM L_w\frac{E_c}{N_0}\sigma\beta_{\mathsf{RT}}^2.
    \label{eq_sensing_snr}
\end{align}
\end{lemma}
\begin{proof}
After removal of the known symbols ($d_m^*d_m=1$, noise statistics unchanged), the per-chip SNR of the direct return is $E_c|\xi_0|^2/N_0$. Circular correlation with the length-$N$ MLS coherently sums $N$ chips at the peak $R_s[0]=1$ while reducing the noise variance to $N_0/N$, yielding a range-compression gain of $N$. The slow-time DFT coherently integrates $M$ pulses with window efficiency $L_w$, contributing a further gain factor $ML_w$. Multiplying these gains yields the first line of \eqref{eq_sensing_snr}. Substituting $|\xi_0|^2=\sigma\beta_{\mathsf{RT}}^2$ from Lemma~\ref{lem_returns} gives the final expression, which accounts for the direct echo only. 
\end{proof}

This result in Lemma \ref{lem_sensing_snr} underscores a fundamental sensing--communication tradeoff: the RIS configuration optimized for communication gain not only provides no sensing advantage but causes the RIS-assisted returns $\xi_1,\xi_2,\xi_3$ to appear as clutter noise rather than coherent gain, as illustrated in Fig.~\ref{fig_rdmap}.

\begin{remark}
Since the direct path is deterministic and the RCS is fixed, $\chi_0$ is constant across noise realizations. By Remark~\ref{rem_noncoherent}, the RIS-assisted returns are mis-phased and do not enhance $\chi_0$; the sensing branch therefore performs essentially as in the no-RIS case, while the self-return $\xi_2$ adds zero-Doppler clutter.
\end{remark}

\subsection{Communication SNR and BER}

Since one bit occupies all $N$ chips, $E_b=NE_c$, and the post-despreading bit SNR at the UE is
\begin{align}
    \gamma_b
    =
    \frac{E_b|h_{\mathrm{eff}}|^2}{N_0}
    =
    \overline\gamma A_c^2,
    \qquad
    \overline\gamma\triangleq\frac{E_b}{N_0}=N\frac{E_c}{N_0}.
    \label{eq_bit_snr}
\end{align}
Under coherent BPSK with perfect timing and channel knowledge, $P_b(\gamma_b)=Q(\sqrt{2\gamma_b})$, and the fading-averaged BER is $\overline P_b=\mathbb E_{\gamma_b}[Q(\sqrt{2\gamma_b})]$.

\subsubsection{Amplitude Moments}

From \eqref{eq_effective_amplitude}, $A_c=X_0+\sum_{l=1}^{L}X_l$ with $X_0\triangleq\sqrt{\beta_{\mathsf{RU}}}\,a_{\mathsf{RU}}$ and $X_l\triangleq\kappa_l\sqrt{\beta_{\mathsf{RI}}\beta_{\mathsf{IU}}}\,a_{\mathsf{RI},l}a_{\mathsf{IU},l}$, all nonnegative and mutually independent. For $a\sim\operatorname{Nakagami}(m,1)$, define the normalized $r$-th raw moment~\cite[cf.~Eq.~(2.169)]{Barkat2005}
\begin{align}
    \mathcal M_r(m)
    \triangleq
    \frac{\Gamma(m+\tfrac r2)}{\Gamma(m)}\,m^{-r/2},
\end{align}
so the component moments are $\mu_{0,r}=\mathbb{E}[X_0^r]=\beta_{\mathsf{RU}}^{r/2}\mathcal M_r(m_{\mathsf{RU}})$ and $\mu_{l,r}=\mathbb{E}[X_l^r]=\kappa_l^r(\beta_{\mathsf{RI}}\beta_{\mathsf{IU}})^{r/2}\mathcal M_r(m_{\mathsf{RI}})\mathcal M_r(m_{\mathsf{IU}})$.

\begin{lemma}[Exact amplitude-moment recursion]
\label{lem_moments}
Let $S_q=X_0+\sum_{l=1}^{q}X_l$ and $\mathcal A_{q,r}\triangleq\mathbb E[S_q^r]$. Then, for $r=1,\ldots,4$,
\begin{align}
    \mathcal A_{q,r}
    =
    \sum_{t=0}^{r}\binom{r}{t}\mathcal A_{q-1,r-t}\mu_{q,t},
    \label{eq_moment_recursion}
\end{align}
with $\mathcal A_{0,r}=\mu_{0,r}$ and $\mathcal A_{q,0}=1$, so that $\mathbb E[A_c^2]=\mathcal A_{L,2}$ and $\mathbb E[A_c^4]=\mathcal A_{L,4}$.
\end{lemma}
\begin{proof}
Write $S_q=S_{q-1}+X_q$ with $S_{q-1}$ and $X_q$ independent. The binomial theorem gives $S_q^r=\sum_{t=0}^{r}\binom{r}{t}S_{q-1}^{\,r-t}X_q^{\,t}$, and taking expectations with $\mathbb E[S_{q-1}^{\,r-t}X_q^{\,t}]=\mathcal A_{q-1,r-t}\,\mu_{q,t}$ by independence yields \eqref{eq_moment_recursion}. The initialization follows from $S_0=X_0$, and $\mathbb E[X_q^t]=\mu_{q,t}$ from the Nakagami raw moments.
\end{proof}

\subsubsection{Closed-Form BER}

Matching the mean and variance of $\gamma_b$ obtained from Lemma~\ref{lem_moments} to a Gamma distribution then yields a closed-form average BER.

\begin{theorem}[Moment-matched average BER]
\label{thm_ber}
Let $\mu_\gamma=\overline\gamma\,\mathcal A_{L,2}$ and $\sigma_\gamma^2=\overline\gamma^2(\mathcal A_{L,4}-\mathcal A_{L,2}^2)$, and approximate $\gamma_b$ by a Gamma random variable with shape $k_\gamma=\mu_\gamma^2/\sigma_\gamma^2$ and scale $\vartheta_\gamma=\sigma_\gamma^2/\mu_\gamma$. Then the average UE BER admits the closed form
\begin{align}
    \overline P_b
    \approx
    \frac1\pi
    \int_0^{\pi/2}
    \left(1+\frac{\vartheta_\gamma}{\sin^2\theta}\right)^{-k_\gamma}
    d\theta.
    \label{eq_ber_gamma}
\end{align}
\end{theorem}
\begin{proof}
Since $\gamma_b=\overline\gamma A_c^2$, Lemma~\ref{lem_moments} gives $\mu_\gamma=\mathbb E[\gamma_b]$ and $\sigma_\gamma^2=\operatorname{var}(\gamma_b)$; matching these to a Gamma$(k_\gamma,\vartheta_\gamma)$ density yields the stated $k_\gamma,\vartheta_\gamma$. Using Craig's representation $Q(x)=\frac1\pi\int_0^{\pi/2}\exp\!\big(-\tfrac{x^2}{2\sin^2\theta}\big)\,d\theta$~\cite[Eq.~(6.44)]{Goldsmith2005},
\begin{align}
    \overline P_b
    =
    \mathbb E_{\gamma_b}\!\left[Q(\sqrt{2\gamma_b})\right]
    =
    \frac1\pi\int_0^{\pi/2}
    \mathbb E_{\gamma_b}\!\left[e^{-\gamma_b/\sin^2\theta}\right]d\theta.
\end{align}
The inner expectation is the Gamma MGF $M_{\gamma_b}(s)=(1-\vartheta_\gamma s)^{-k_\gamma}$ evaluated at $s=-1/\sin^2\theta$~\cite[cf.~Eq.~(2.105)]{Barkat2005}, which gives $(1+\vartheta_\gamma/\sin^2\theta)^{-k_\gamma}$ and hence \eqref{eq_ber_gamma}.
\end{proof}

\begin{remark}
Equation \eqref{eq_ber_gamma} is a moment-matched approximation, not an exact distributional result. As $L$ grows, $A_c$ aggregates more independent terms and its squared envelope is increasingly well approximated by a Gamma law, so the approximation tightens (cf. Corollary~\ref{cor_l2_gain}).
\end{remark}

In the absence of the RIS, the effective amplitude reduces to a single Nakagami term and the BER admits an exact closed form.

\begin{corollary}[Exact no-RIS baseline BER]
\label{cor_noris}
For $L=0$, $\gamma_{b,0}=\overline\gamma\beta_{\mathsf{RU}}a_{\mathsf{RU}}^2$ is exactly Gamma distributed and the average BER is
\begin{align}
    \overline P_b^{\mathrm{dir}}
    =
    \frac1\pi
    \int_0^{\pi/2}
    \left(1+\frac{\overline\gamma\beta_{\mathsf{RU}}}{m_{\mathsf{RU}}\sin^2\theta}\right)^{-m_{\mathsf{RU}}}
    d\theta.
    \label{eq_ber_noris}
\end{align}
\end{corollary}
\begin{proof}
For $L=0$, $A_c=\sqrt{\beta_{\mathsf{RU}}}\,a_{\mathsf{RU}}$. A Nakagami-$m$ amplitude with unit spread satisfies $a_{\mathsf{RU}}^2\sim\operatorname{Gamma}(m_{\mathsf{RU}},1/m_{\mathsf{RU}})$, so $\gamma_{b,0}\sim\operatorname{Gamma}(m_{\mathsf{RU}},\overline\gamma\beta_{\mathsf{RU}}/m_{\mathsf{RU}})$. Applying the Craig--MGF argument of Theorem~\ref{thm_ber} with the exact shape $m_{\mathsf{RU}}$ and scale $\overline\gamma\beta_{\mathsf{RU}}/m_{\mathsf{RU}}$ gives \eqref{eq_ber_noris}.
\end{proof}

\section{Numerical Results}
\label{sec_results}

We corroborate the analysis with Monte-Carlo simulation; the key settings are
summarized in Table~\ref{tab_params}. The communication links follow
Nakagami-$m$ fading, the sensing links are deterministic line-of-sight, and the
RIS is configured by the communication-centric phases \eqref{eq_comm_ris_phase}.

\begin{table}[t]
\centering
\footnotesize
\renewcommand{\arraystretch}{1.0}
\caption{Key Simulation Settings}
\label{tab_params}
\begin{tabular}{@{}ll@{}}
\toprule
Parameter & Value \\
\midrule
\multicolumn{2}{@{}l}{\textit{Waveform and frame}}\\
Carrier frequency $f_c$ (wavelength $\lambda$)         & $24$~GHz ($12.5$~mm)\\
MLS order / length $k$ / $N$                            & $7$ / $127$\\
Pulses per CPI $M$ / pilots $M_{\mathrm{pl}}$           & $64$ / $8$\\
Chip rate $f_s=$ bandwidth $B$                          & $20$~MHz\\
Chip / PRI / CPI duration                              & $50$~ns / $6.35\,\mu$s / $406\,\mu$s\\
Range / velocity resolution                            & $7.49$~m / $15.4$~m/s\\
\midrule
\multicolumn{2}{@{}l}{\textit{Channels}}\\
Reference distance $d_0$                               & $100$~m\\
Path-loss exp. $\alpha_{\mathsf{RU}}$ / $\alpha_{\mathsf{RI}},\alpha_{\mathsf{IU}}$ / $\alpha_{\mathsf{RT}},\alpha_{\mathsf{IT}}$ & $2.2$ / $2.0$ / $2.0$\\
Ref. gains $\Psi_{\mathsf{RU}}$ / $\Psi_{\mathsf{RI}},\Psi_{\mathsf{IU}}$ & $0.5$ / $0.02$\\
Nakagami $m_{\mathsf{RU}},m_{\mathsf{RI}},m_{\mathsf{IU}}$ & $2,\,3,\,2$\\
Target RCS $\sigma$ / atmos. atten.                    & $1$~m$^2$ / $0.01$~dB/km\\
Antenna gains $G,G_{\mathsf I}$ / system loss $L_s$    & $30,5$~dBi / $5$~dB\\
\midrule
\multicolumn{2}{@{}l}{\textit{RIS and detection}}\\
RIS elements $L$ / spacing / $\kappa_l$                & $16$ / $\lambda/2$ / $0.85$\\
CFAR $P_{\mathrm{FA}}$ / train., guard cells           & $10^{-4}$ / $4,1$\\
\midrule
\multicolumn{2}{@{}l}{\textit{Geometry [m] (see Fig.~\ref{fig_topology})}}\\
$\mathsf R,\mathsf I,\mathsf U,\mathsf T$ positions     & $(0,0),(150,30),$\\
                                                        & $(300,0),(500,200)$\\
Target velocity $\mathbf v_{\mathsf T}$                 & $(24,-24)$~m/s\\
\bottomrule
\end{tabular}
\end{table}

\begin{figure}[t]
    \centering
    \includegraphics[width=\columnwidth]{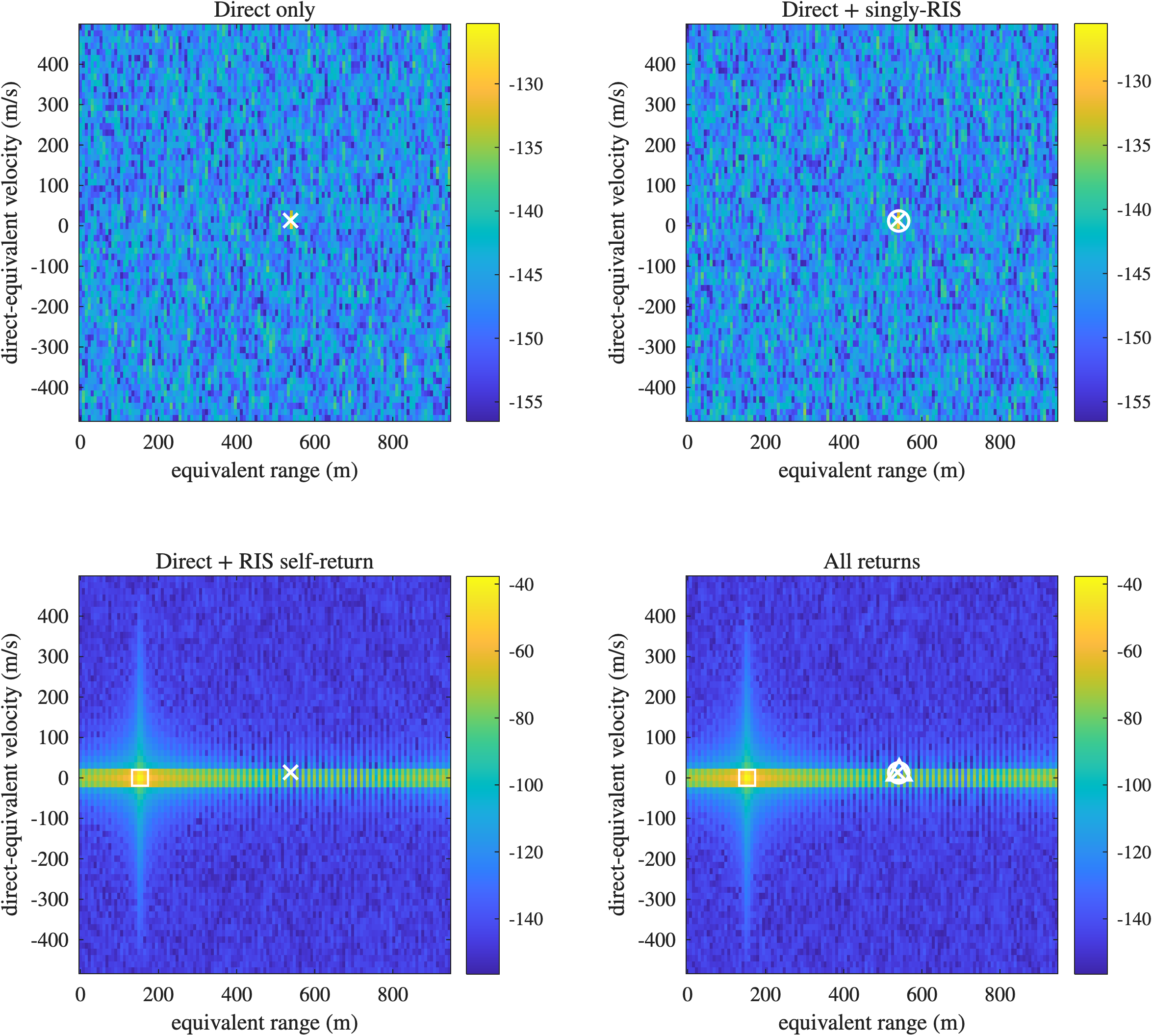}
    \caption{Range--Doppler power maps for the four reflection mechanisms.}
    \label{fig_rdmap}
\end{figure}

Fig.~\ref{fig_rdmap} shows the range--Doppler map obtained from the composite
radar observation \eqref{eq_radar_echo} for the four reflection mechanisms of
Lemma~\ref{lem_returns}. The direct target return produces a sharp peak
(marked \texttt{x}) at its range--Doppler cell and is unchanged by the RIS,
while the communication-centric RIS configuration leaves the RIS-assisted
returns mis-phased: the static self-return $\xi_2$ appears as a zero-Doppler
clutter ridge and the singly/doubly RIS-assisted returns $\xi_1,\xi_3$ do not
coherently reinforce the target. This confirms the structural prediction of
Lemma~\ref{lem_returns} and Remark~\ref{rem_noncoherent}.

\begin{figure}[t]
    \centering
    \includegraphics[width=\columnwidth]{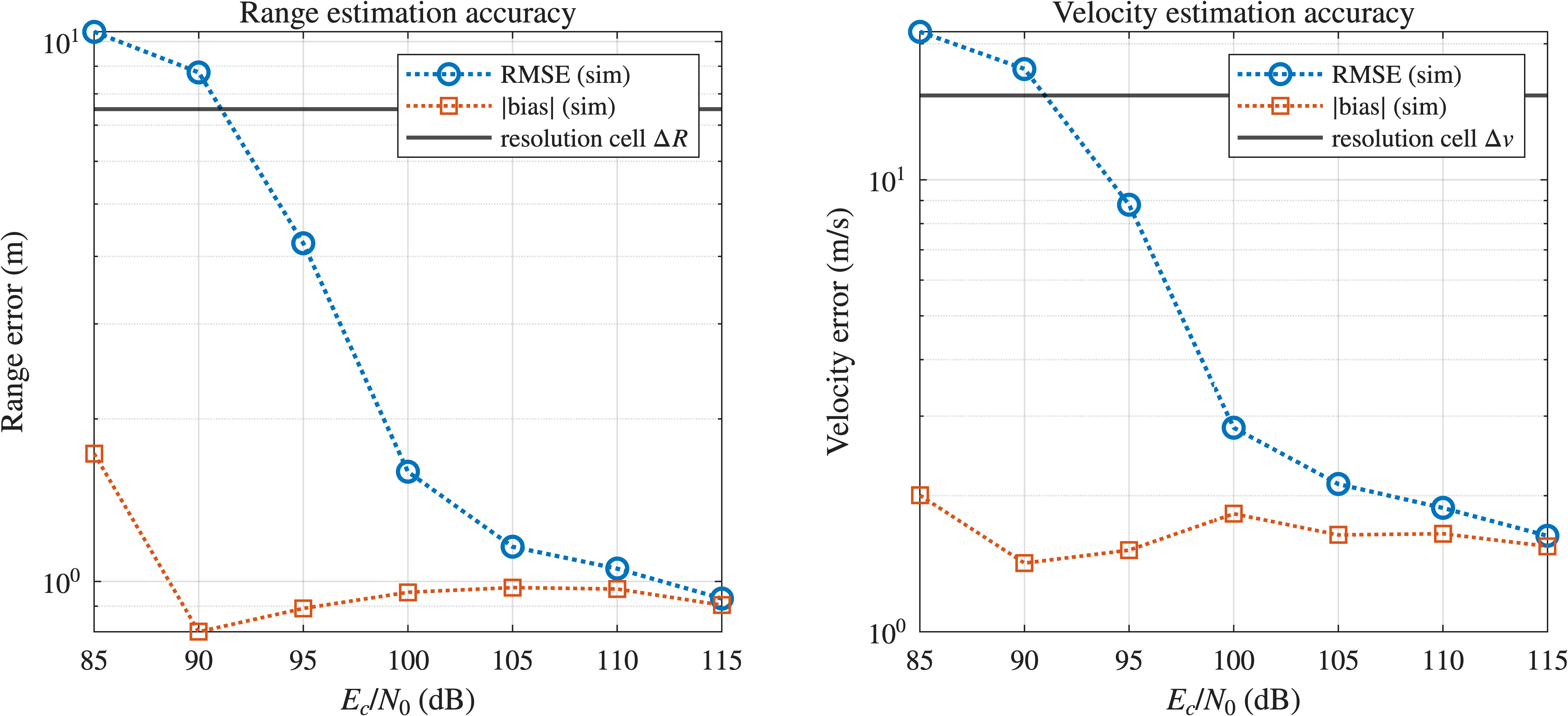}
    \caption{Direct-path range and velocity estimation accuracy versus $E_c/N_0$.}
    \label{fig_sensing}
\end{figure}

Fig.~\ref{fig_sensing} reports the root-mean-square error (RMSE) and bias of
the direct-path range and velocity estimates as a function of $E_c/N_0$, using
the direct-only return with a rectangular slow-time window so that the
accuracy reflects the intrinsic capability of the dual-function waveform. The RMSE is large at low $E_c/N_0$, where the
range--Doppler peak competes with noise, and decreases monotonically, dropping
below the resolution cells $\Delta R=c_0T_c/2$ and $\Delta v=\lambda/(2MT_{\mathrm{PRI}})$
at high $E_c/N_0$, i.e., the estimator attains sub-resolution accuracy through
sub-bin peak interpolation. Across the sweep the bias magnitude remains far
below the RMSE, so the error is dominated by zero-mean random spread rather
than a systematic offset, confirming that the estimator is essentially
unbiased; the small high-$E_c/N_0$ floor is the residual sub-bin
interpolation/grid bias.

\begin{figure}[t]
    \centering
    \includegraphics[width=\columnwidth]{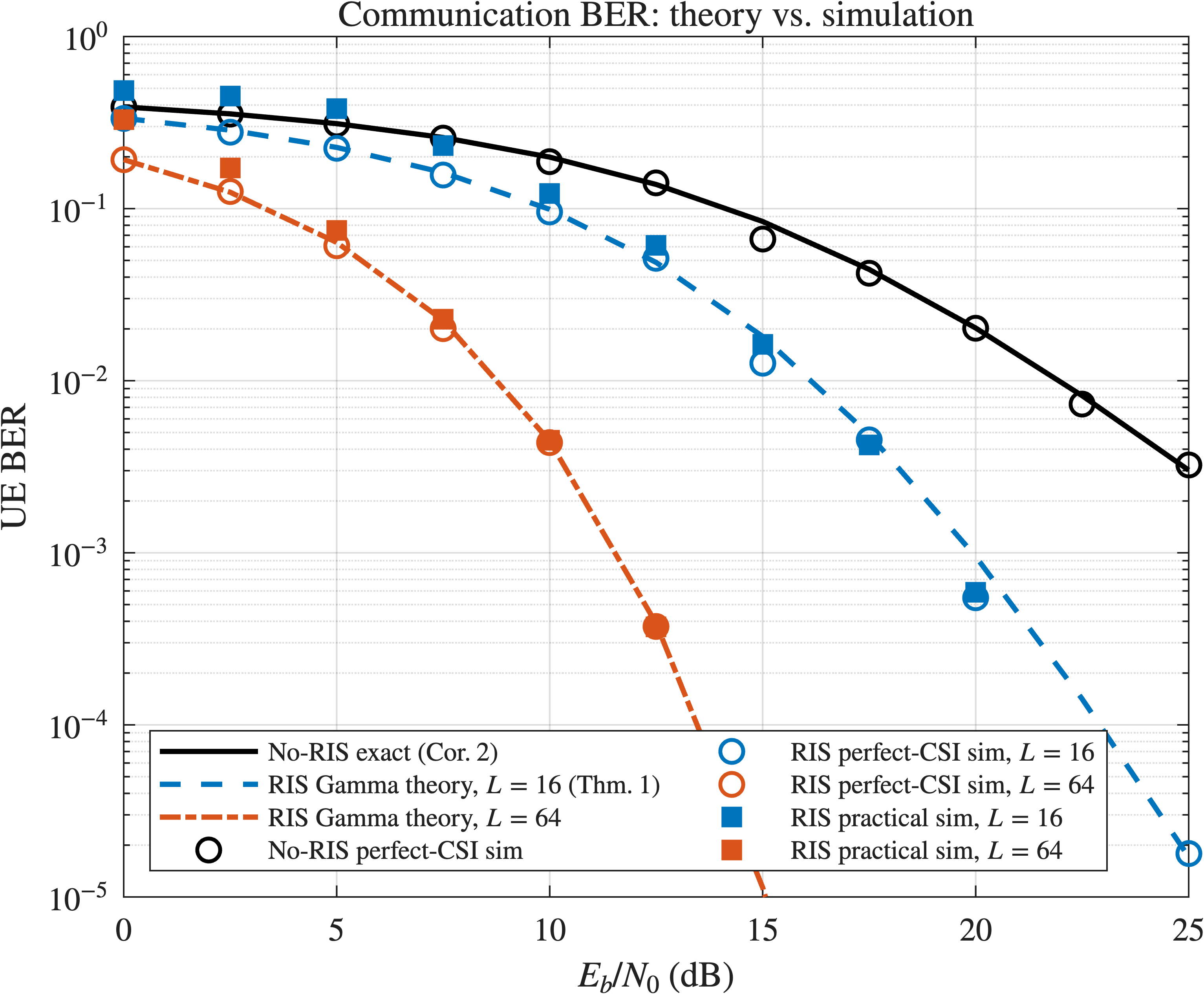}
    \caption{Average UE BER versus $E_b/N_0$ for $L\in\{16,64\}$.}
    \label{fig_ber}
\end{figure}

Fig.~\ref{fig_ber} compares the average UE BER from
Theorem~\ref{thm_ber} and Corollary~\ref{cor_noris} against Monte-Carlo
simulation for $L\in\{16,64\}$. The moment-matched Gamma approximation \eqref{eq_ber_gamma}
closely follows both the perfect-CSI bit simulation and the practical
pilot-based receiver, and the exact no-RIS baseline \eqref{eq_ber_noris}
agrees with its simulation. The RIS-assisted link attains a large $E_b/N_0$ advantage over the
direct link, and increasing $L$ from $16$ to $64$ shifts the BER curve left
by the $\approx 12$~dB predicted by the $20\log_{10}L$ array gain of
Corollary~\ref{cor_l2_gain}.

\begin{remark}[Estimation/synchronization penalty]
\label{rem_estimation_penalty}
For each RIS size, Fig.~\ref{fig_ber} reports two simulated curves: a
perfect-CSI receiver with genie channel knowledge and the
practical receiver of Section~\ref{sec_receiver} that acquires timing and
estimates $h_{\mathrm{eff}}$ from a short $M_{\mathrm{pl}}$-symbol pilot prefix
by least squares. The gap between them is the
estimation/synchronization penalty incurred by replacing genie knowledge with
pilot-based estimation. This gap is small and concentrated at low $E_b/N_0$,
where the short-pilot least-squares estimate of $h_{\mathrm{eff}}$ is noisiest.
The resulting three-way agreement, practical $\approx$ perfect-CSI $\approx$
moment-matched Gamma theory \eqref{eq_ber_gamma}, validates both the UE
receiver of Section~\ref{sec_receiver} and the closed-form BER of
Theorem~\ref{thm_ber}.
\end{remark}

\begin{figure}[t]
    \centering
    \includegraphics[width=\columnwidth]{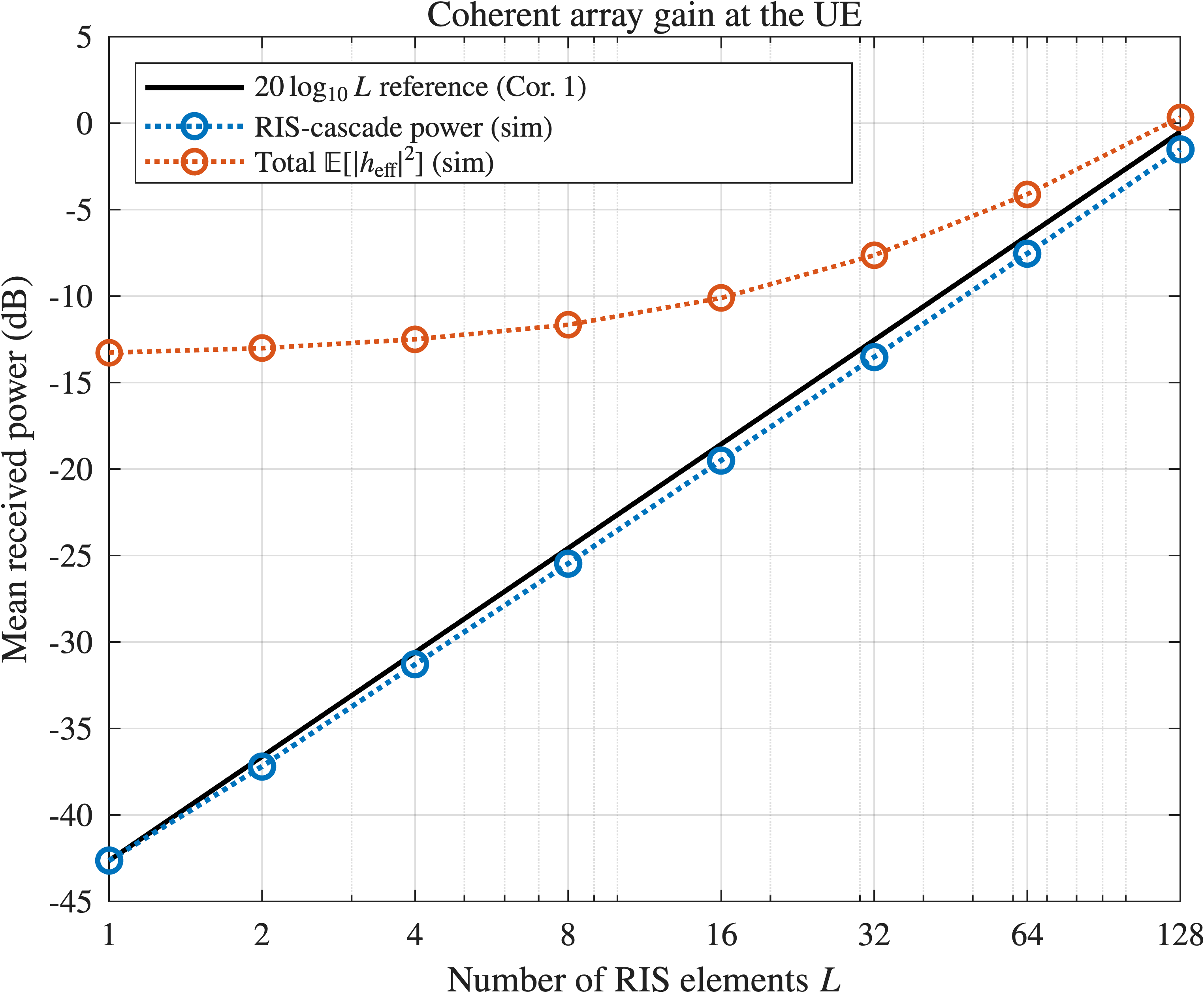}
    \caption{Mean received power at the UE versus $L$.}
    \label{fig_gain}
\end{figure}

Fig.~\ref{fig_gain} reports the mean received power $\mathbb E[|h_{\mathrm{eff}}|^2]$
as a function of $L$. The RIS-cascade component grows at $19.9$~dB/decade,
matching the $20\log_{10}L$ slope predicted by Corollary~\ref{cor_l2_gain},
and the total effective power transitions into the $\Theta(L^2)$ regime once
the RIS term dominates the direct path. This validates the $L^2$ coherent
array gain delivered to the UE by the communication-centric RIS
configuration.

\section{Conclusion}
\label{sec_conclusion}

We considered a monostatic RIS-assisted ISAC system using MLS-based PMCW waveforms over Nakagami-$m$ fading. We proposed a communication-centric RIS configuration co-phasing onto the direct radar--UE path and a dual-function waveform embedding BPSK bits via pulse polarity modulation. We derived closed-form BER expressions using a moment-matched Gamma approximation and sensing SNR for the direct echo. We showed that the RIS configuration achieves $L^2$ coherent array gain at the UE with robust pilot-based channel estimation; however, RIS-assisted returns become mis-phased at the radar, acting as clutter while preserving accurate range-Doppler estimation from the direct echo.

\bibliographystyle{IEEEtran}
\bibliography{references}
\end{document}